\documentclass[10pt,twocolumn,english,conference]{IEEEtran}
\usepackage[T1]{fontenc}
\usepackage[latin1]{inputenc}
\usepackage{geometry}
\usepackage{verbatim}
\usepackage{subfigure}
\usepackage{amsmath}
\usepackage{graphicx}
\usepackage{amssymb}

\makeatletter
  \newtheorem{thm}{Theorem}
  \newtheorem{example}{Example}
  \newtheorem{cor}{Corollary}

\usepackage{cite}

\usepackage{babel}
\makeatother
\begin{document}

\title{Quantization Bounds on Grassmann Manifolds of Arbitrary Dimensions
and MIMO Communications with Feedback$^{^{*}}$}

\author{\authorblockN{Wei Dai, Youjian Liu} \authorblockA{Dept. of Electrical and Computer Eng.\\ University of Colorado at Boulder\\ Boulder CO 80303, USA\\ Email: dai@colorado.edu, eugeneliu@ieee.org} \and \authorblockN{Brian Rider} \authorblockA{Math Department\\ University of Colorado at Boulder\\ Boulder CO 80303, USA\\ Email: brider@euclid.colorado.edu} }

\maketitle

\begin{abstract}
This paper considers the quantization problem on the Grassmann manifold
with dimension $n$ and $p$. The unique contribution is the derivation
of a closed-form formula for the volume of a metric ball in the Grassmann
manifold when the radius is sufficiently small. This volume formula
holds for Grassmann manifolds with arbitrary dimension $n$ and $p$,
while previous results are only valid for either $p=1$ or a fixed
$p$ with asymptotically large $n$. Based on the volume formula,
the Gilbert-Varshamov and Hamming bounds for sphere packings are obtained.
Assuming a uniformly distributed source and a distortion metric based
on the squared chordal distance, tight lower and upper bounds are
established for the distortion rate tradeoff. Simulation results match
the derived results. As an application of the derived quantization
bounds, the information rate of a Multiple-Input Multiple-Output (MIMO)
system with finite-rate channel-state feedback is accurately quantified
for arbitrary finite number of antennas, while previous results are
only valid for either Multiple-Input Single-Output (MISO) systems
or those with asymptotically large number of transmit antennas but
fixed number of receive antennas.
\end{abstract}
\renewcommand{\thefootnote}{\fnsymbol{footnote}} \footnotetext[1]{This work is partially supported by the Junior Faculty Development Award, University of Colorado at Boulder.} \renewcommand{\thefootnote}{\arabic{footnote}} \setcounter{footnote}{0}

\section{Introduction\label{sec:Introduction}}




The \emph{Grassmann manifold} $\mathcal{G}_{n,p}\left(\mathbb{L}\right)$
is the set of all $p$-dimensional planes (through the origin) of
the $n$-dimensional Euclidean space $\mathbb{L}^{n}$, where $\mathbb{L}$
is either $\mathbb{R}$ or $\mathbb{C}$. It forms a compact Riemann
manifold of real dimension $\beta p\left(n-p\right)$, where $\beta=1/2$
when $\mathbb{L}=\mathbb{R}/\mathbb{C}$ respectively. The Grassmann
manifold provides a useful analysis tool for multi-antenna communications
(also known as Multiple-Input Multiple-Output (MIMO) communication
systems. For non-coherent MIMO systems, sphere packings on the $\mathcal{G}_{n,p}\left(\mathbb{L}\right)$
can be viewed as a generalization of spherical codes \cite{Urbanke_IT01_Signal_Constellations,Tse_IT02_Communication_on_Grassmann_Manifold,Barg_IT02_Bounds_Grassmann_Manifold}.
For MIMO systems with finite rate channel state feedback, the quantization
of beamforming matrices is related to the quantization on the Grassmann
manifold \cite{Sabharwal_IT03_Beamforming_MIMO,Love_IT03_Grassman_Beamforming_MIMO,Dai_05_Power_onoff_strategy_design_finite_rate_feedback}.

The basic quantization problems addressed in this paper are the sphere
packing bounds and distortion rate tradeoff. A quantization is a mapping
from the $\mathcal{G}_{n,p}\left(\mathbb{L}\right)$ into a subset
of the $\mathcal{G}_{n,p}\left(\mathbb{L}\right)$, known as the code
$\mathcal{C}$. Define $\delta\triangleq\delta\left(\mathcal{C}\right)$
as the minimum distance between any two elements in $\mathcal{C}$.
The sphere packing bounds relate the size of a code and a given minimum
distance $\delta$. Assuming a randomly distributed source on the
$\mathcal{G}_{n,p}\left(\mathbb{L}\right)$ and a distortion metric,
the distortion rate tradeoff is described by either the minimum expected
distortion achievable for a given code size (distortion rate function)
or the minimum code size required to achieve a particular expected
distortion (rate distortion function).

%
{}




For the sake of applications\cite{Sabharwal_IT03_Beamforming_MIMO,Love_IT03_Grassman_Beamforming_MIMO,Dai_05_Power_onoff_strategy_design_finite_rate_feedback},
the projection Frobenius metric (i.e. \emph{chordal distance}) is
employed throughout the paper although the corresponding analysis
is also applicable to the geodesic metric \cite{Barg_IT02_Bounds_Grassmann_Manifold}.
For any two planes $P,Q\in\mathcal{G}_{n,p}\left(\mathbb{L}\right)$,
the principle angles and the chordal distance between $P$ and $Q$
are defined as follows. Let $\mathbf{u}_{1}\in P$ and $\mathbf{v}_{1}\in Q$
be the unit vectors such that $\left|\mathbf{u}_{1}^{\dagger}\mathbf{v}_{1}\right|$
is maximal. Inductively, let $\mathbf{u}_{i}\in P$ and $\mathbf{v}_{i}\in Q$
be the unit vectors such that $\mathbf{u}_{i}^{\dagger}\mathbf{u}_{j}=0$
and $\mathbf{v}_{i}^{\dagger}\mathbf{v}_{j}=0$ for all $1\leq j<i$
and $\left|\mathbf{u}_{i}^{\dagger}\mathbf{v}_{i}\right|$ is maximal.
The principle angles are defined as $\theta_{i}=\arccos\left|\mathbf{u}_{i}^{\dagger}\mathbf{v}_{i}\right|$
for $i=1,\cdots,n$ \cite{James_54_Normal_Multivariate_Analysis_Orthogonal_Group,Conway_96_PackingLinesPlanes}.
The chordal distance between $P$ and $Q$ is defined as \[
d_{c}\left(P,Q\right)\triangleq\sqrt{\sum_{i=1}^{p}\sin^{2}\theta_{i}}.\]

The invariant measure on the $\mathcal{G}_{n,p}\left(\mathbb{L}\right)$
is defined as follows. Let $O\left(n\right)/U\left(n\right)$ be the
group of $n\times n$ orthogonal/unitary matrices respectively. Let
$\mathbf{A}\in O\left(n\right)/U\left(n\right)$ and $\mathbf{B}\in O\left(n\right)/U\left(n\right)$
when $\mathbb{L}=\mathbb{R}/\mathbb{C}$ respectively. An invariant
measure $\mu$ on the $\mathcal{G}_{n,p}\left(\mathbb{L}\right)$
satisfies, for any measurable set $\mathcal{M}\subset\mathcal{G}_{n,p}\left(\mathbb{L}\right)$
and arbitrarily chosen $\mathbf{A}$ and $\mathbf{B}$, \[
\mu\left(\mathbf{A}\mathcal{M}\right)=\mu\left(\mathcal{M}\right)=\mu\left(\mathcal{M}\mathbf{B}\right).\]
 The invariant measure defines the uniform distribution on the $\mathcal{G}_{n,p}\left(\mathbb{L}\right)$
\cite{James_54_Normal_Multivariate_Analysis_Orthogonal_Group}.

With a metric and a measure defined on the $\mathcal{G}_{n,p}\left(\mathbb{L}\right)$,
there are several bounds well known for sphere packings. Let $\delta$
be the minimum distance between any two elements of a code $\mathcal{C}$
and $B\left(\delta\right)$ be the metric ball of radius $\delta$
in the $\mathcal{G}_{n,p}\left(\mathbb{L}\right)$. If $K$ is any
number such that $K\mu\left(B\left(\delta\right)\right)<1$, then
there exists a code $\mathcal{C}$ of size $K+1$ and minimum distance
$\delta$. This principle is called as the \emph{Gilbert-Varshamov}
lower bound \cite{Barg_IT02_Bounds_Grassmann_Manifold}, i.e. \begin{equation}
\left|\mathcal{C}\right|>\frac{1}{\mu\left(B\left(\delta\right)\right)}.\label{eq:GV_bd}\end{equation}
On the other hand, $\left|\mathcal{C}\right|\mu\left(B\left(\delta/2\right)\right)\leq1$
for any code $\mathcal{C}$. The \emph{Hamming} upper bound captures
this fact as\cite{Barg_IT02_Bounds_Grassmann_Manifold}\begin{equation}
\left|\mathcal{C}\right|\leq\frac{1}{\mu\left(B\left(\delta/2\right)\right)}.\label{eq:Hamming_bd}\end{equation}
These two bounds relate the code size and a given minimum distance
$\delta$.

Distortion rate function gives another important property of quantization.
Assume that $Q$ is a random plane uniformly distributed on the $\mathcal{G}_{n,p}\left(\mathbb{L}\right)$
and a distortion metric defined by the squared chordal distance $d_{c}^{2}$.
The average distortion of a given $\mathcal{C}$ is \begin{equation}
D\left(\mathcal{C}\right)\triangleq E_{Q}\left[\underset{P\in\mathcal{C}}{\min}\; d_{c}^{2}\left(P,Q\right)\right].\label{eq:distortion_def}\end{equation}
The distortion rate function gives the minimum average distortion
for a given codebook size $K$, i.e.\begin{equation}
D^{*}\left(K\right)=\underset{\mathcal{C}:\left|\mathcal{C}\right|=K}{\inf}\; D\left(\mathcal{C}\right).\label{eq:rate_distortion_def}\end{equation}


There are several papers addressing quantization problems in the Grassmann
manifold. The exact volume formula for a $B\left(\delta\right)$ in
the $\mathcal{G}_{n,p}\left(\mathbb{C}\right)$ where $p=1$ is derived
in \cite{Sabharwal_IT03_Beamforming_MIMO}. An asymptotic volume formula
for a $B\left(\delta\right)$ in the $\mathcal{G}_{n,p}\left(\mathbb{L}\right)$,
where $p\geq1$ is fixed and $n$ approaches infinity, is derived
in \cite{Barg_IT02_Bounds_Grassmann_Manifold}. Based on those volume
formulas, the corresponding sphere packing bounds are developed in
\cite{Love_IT03_Grassman_Beamforming_MIMO,Barg_IT02_Bounds_Grassmann_Manifold}.
Besides the sphere packing bounds, the rate distortion tradeoff is
also treated in \cite{Heath_ICASSP05_Quantization_Grassmann_Manifold},
where approximations to the distortion rate function are derived by
the sphere packing bounds. However, the derived approximations are
based on the volume formulas \cite{Barg_IT02_Bounds_Grassmann_Manifold,Sabharwal_IT03_Beamforming_MIMO}
only valid for some special choices of $n$ and $p$, i.e. either
$p=1$ or fixed $p\geq1$ with asymptotic large $n$.

This paper derives quantization bounds for the Grassmann manifold
with arbitrary $n$ and $p$ when the code size is large. An explicit
volume formula for a metric ball in the $\mathcal{G}_{n,p}\left(\mathbb{L}\right)$
is derived when the radius is sufficiently small. Based on the derived
volume formula, the sphere packing bounds are obtained. The distortion
rate tradeoff is also characterized by establishment of tight lower
and upper bounds. Simulation results match the derived bounds. As
an application of the derived quantization bounds, the information
rate of a MIMO system with finite rate channel state feedback is accurately
quantified for abitrary finite number of antennas for the first time,
while previous results are only valid for either Multiple-Input Single-Output
(MISO) systems or those with asymptotically large number of transmit
antennas but fixed number of receive antennas.

\section{Metric Balls in the $\mathcal{G}_{n,p}\left(\mathbb{L}\right)$\label{sec:Spheres}}



In this section, an explicit volume formula for a metric ball $B\left(\delta\right)$
in the $\mathcal{G}_{n,p}\left(\mathbb{L}\right)$ is derived. The
volume formula is essential for the quantization bounds in Section
\ref{sec:Quantization-Bounds}. 

The volume calculation depends on the relationship between the measure
and the metric defined on the $\mathcal{G}_{n,p}\left(\mathbb{L}\right)$.
For the invariant measure $\mu$ and the chordal distance $d_{c}$,
the volume of a metric ball $B\left(\delta\right)$ can be calculated
by\begin{equation}
\mu\left(B\left(\delta\right)\right)=\underset{\underset{\frac{\pi}{2}\geq\theta_{1}\geq\cdots\geq\theta_{p}\geq0}{\sqrt{\sum_{i=1}^{p}\sin^{2}\theta_{i}}\leq\delta}}{\int\cdots\int}\; d\mu_{\mathbf{\theta}},\label{eq:actual_volume}\end{equation}
where $\theta_{1},\cdots,\theta_{p}$ are the principle angles and
the differential form $d\mu_{\mathbf{\theta}}$ is given in \cite{James_54_Normal_Multivariate_Analysis_Orthogonal_Group,Adler_2001_Integrals_Grassmann}. 

The following theorem expresses the volume formula as an exponentiation
of the radius $\delta$.

\begin{thm}
\label{thm:Volume_formula}Let $B\left(\delta\right)$ be a ball of
radius $\delta$ in $\mathcal{G}_{n,p}\left(\mathbb{L}\right)$. When
$\delta\leq1$, \begin{equation}
\mu\left(B\left(\delta\right)\right)=\left\{ \begin{array}{ll}
c_{n,p,\beta}\delta^{p\left(n-p\right)}\left(1+o\left(\delta\right)\right) & \mathrm{if}\;\mathbb{L}=\mathbb{R}\\
c_{n,p,\beta}\delta^{2p\left(n-p\right)} & \mathrm{if}\;\mathbb{L}=\mathbb{C}\end{array}\right.,\label{eq:simplified_volume_formula}\end{equation}
where $\beta=1/2$ when $\mathbb{L}=\mathbb{R}/\mathbb{C}$ respectively
and $c_{n,p,\beta}$ is a constant determined by $n$, $p$ and $\beta$.
When $\mathbb{L}=\mathbb{C}$, $c_{n,p,2}$ can be explicitly calculated\begin{equation}
c_{n,p,2}=\left\{ \begin{array}{ll}
\frac{1}{\left(np-p^{2}\right)!}\prod_{i=1}^{p}\frac{\left(n-i\right)!}{\left(p-i\right)!} & \mathrm{if}\;0<p\leq\frac{n}{2}\\
\frac{1}{\left(np-p^{2}\right)!}\prod_{i=1}^{n-p}\frac{\left(n-i\right)!}{\left(n-p-i\right)!} & \mathrm{if}\;\frac{n}{2}\leq p\leq n\end{array}\right..\label{eq:constant-complex-manifold}\end{equation}
When $\mathbb{L}=\mathbb{R}$, $c_{n,p,1}$ is given by \begin{eqnarray}
 &  & c_{n,p,1}=\nonumber \\
 &  & \left\{ \begin{array}{ll}
\frac{V_{n,p,1}}{2^{p}}\underset{\underset{x_{1}\geq\cdots\geq x_{p}\geq0}{\sum_{i=1}^{p}x_{i}\leq1}}{\int\cdots\int}\left[\left|\prod_{i<j}^{p}\left(x_{i}-x_{j}\right)\right|\right.\\
\quad\quad\left.\prod_{i=1}^{p}\left(x_{i}^{\frac{1}{2}\left(n-2p+1\right)-1}dx_{i}\right)\right] & \mathrm{if}\;0<p\leq\frac{n}{2}\\
\frac{V_{n,n-p,1}}{2^{n-p}}\underset{\underset{x_{1}\geq\cdots\geq x_{n-p}\geq0}{\sum_{i=1}^{n-p}x_{i}\leq1}}{\int\cdots\int}\left[\left|\prod_{i<j}^{p}\left(x_{i}-x_{j}\right)\right|\right.\\
\quad\quad\left.\prod_{i=1}^{n-p}\left(x_{i}^{\frac{1}{2}\left(2p-n+1\right)-1}dx_{i}\right)\right] & \mathrm{if}\;\frac{n}{2}\leq p\leq n\end{array}\right.,\label{eq:constant_in_volume}\end{eqnarray}
where \[
V_{n,p,1}=\prod_{i=1}^{p}\frac{A^{2}\left(p-i+1\right)A\left(n-p-i+1\right)}{2A\left(n-i+1\right)}\]
and \[
A\left(p\right)=\frac{2\pi^{p/2}}{\Gamma\left(\frac{p}{2}\right)}.\]

\end{thm}

The proof of Theorem \ref{thm:Volume_formula} is not included due
to the length limit.

Theorem \ref{thm:Volume_formula} provides an explicit volume approximation
for real Grassmann manifolds and an exact volume formula for complex
Grassmann manifolds when $\delta\leq1$. Simulations show that this
approximation remains good for relatively large $\delta$ (Fig. \ref{cap:volume_in_Grassmann}).

Theorem \ref{thm:Volume_formula} is consistent with the previous
results in \cite{Sabharwal_IT03_Beamforming_MIMO} and \cite{Barg_IT02_Bounds_Grassmann_Manifold},
which pertain to special choices of $n$ and $p$ and are stated as
follows.

\begin{example}
\label{exa:G_n_1}Consider the volume formula for a $B\left(\delta\right)$
in the $\mathcal{G}_{n,p}\left(\mathbb{C}\right)$ where $p=1$. It
has been shown in \cite{Sabharwal_IT03_Beamforming_MIMO} that \[
\mu\left(B\left(\delta\right)\right)=\delta^{2\left(n-1\right)}.\]
Theorem \ref{thm:Volume_formula} is consistent with it where $\beta=2$
and $c_{n,1,2}=1$. 
\end{example}

\begin{example}
\label{exa:G_asymptotic_n}When $p$ is fixed and $n\rightarrow+\infty$,
the asymptotic volume formula for a $B\left(\delta\right)$ is given
by Barg \cite{Barg_IT02_Bounds_Grassmann_Manifold} as\begin{equation}
\mu\left(B\left(\delta\right)\right)=\left(\frac{\delta}{\sqrt{p}}\right)^{\beta np+o\left(n\right)}.\label{eq:Barg-formula}\end{equation}
On the other hand, Theorem \ref{thm:Volume_formula} contains an asymptotic
formula for $\mathbb{L}=\mathbb{C}$, $\delta\leq1$, fixed $p$ and
asymptotically large $n$ in the form \[
\mu\left(B\left(\delta\right)\right)=\left(\frac{\delta}{\sqrt{p}}\right)^{2p\left(n-p\right)+o\left(n\right)}.\]
This follows from (\ref{eq:constant-complex-manifold}) and Stirling's
approximation. Therefore, Theorem \ref{thm:Volume_formula} is consistent
with Barg's formula (\ref{eq:Barg-formula}).
\end{example}

Importantly though, Theorem \ref{thm:Volume_formula} is distinct
from the previous results of \cite{Sabharwal_IT03_Beamforming_MIMO}
and \cite{Barg_IT02_Bounds_Grassmann_Manifold} in that it holds for
arbitrary $n$ and $p$, $1\leq p\leq n$.

\begin{figure}
\subfigure[Real Grassmann manifolds]{\includegraphics[clip,scale=0.5]{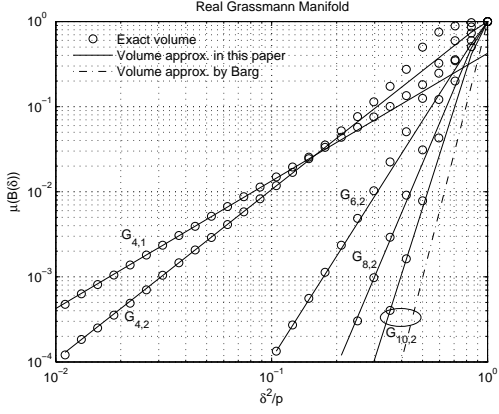}}

\subfigure[Complex Grassmann manifolds]{\includegraphics[clip,scale=0.5]{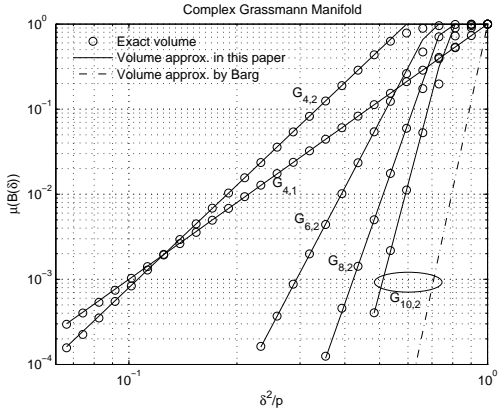}}

\caption{\label{cap:volume_in_Grassmann}The volume of a metric ball in the
Grassmann manifold}
\end{figure}

Fig. \ref{cap:volume_in_Grassmann} compares the exact volume of a
metric ball (\ref{eq:actual_volume}) and the volume evaluated by
(\ref{eq:simplified_volume_formula}). For the volume approximation
$c_{n,p,\beta}\delta^{\beta p\left(n-p\right)}$, the constant $c_{n,p,\beta}$
is calculated either by (\ref{eq:constant-complex-manifold}) if $\mathbb{L}=\mathbb{C}$
or by Monte Carlo numerical integral of (\ref{eq:constant_in_volume})
if $\mathbb{L}=\mathbb{R}$. Simulations show that the volume approximation
is close to the exact volume when the radius of the metric ball is
not large. We also compare our approximation with Barg's approximation
$\left(\delta/\sqrt{p}\right)^{\beta np}$ for $n=10$ and $p=2$
case. Simulations show that the exact volume and Barg's approximation
may not be in the same order while the approximation in this paper
is more accurate.

\section{Quantization Bounds\label{sec:Quantization-Bounds}}





Based on the volume formula given in Theorem \ref{thm:Volume_formula},
the sphere packing bounds are derived and the rate distortion tradeoff
is characterized in this section.

The Gilbert-Varshamov and Hamming bounds on the $\mathcal{G}_{n,p}\left(\mathbb{L}\right)$
are given in the following corollary.

\begin{cor}
\label{cor:packing_bounds}When $\delta$ is sufficiently small, there
exists a code in $\mathcal{G}_{n,p}\left(\mathbb{L}\right)$ with
size $K$ and the minimum distance $\delta$ such that \[
c_{n,p,\beta}^{-1}\delta^{-\beta p\left(n-p\right)}\lesssim K.\]
For any code with the minimum distance $\delta$, \[
K\lesssim c_{n,p,\beta}^{-1}\left(\frac{\delta}{2}\right)^{-\beta p\left(n-p\right)}.\]
\emph{Here and throughtout, the symbol $\lesssim$ indicates that
the inequality holds up to $\left(1+o\left(1\right)\right)$ error.} 
\end{cor}
\begin{proof}
The corollary follows by substituting the volume formula (\ref{eq:simplified_volume_formula})
into (\ref{eq:GV_bd}) and (\ref{eq:Hamming_bd}).
\end{proof}

The distortion rate function is characterized by establishing tight
lower and upper bounds.

\begin{thm}
\label{thm:rate_distortion_bounds}Let $t=\beta p\left(n-p\right)$
be the number of the real dimensions of the Grassmann manifold $\mathcal{G}_{n,p}\left(\mathbb{L}\right)$.
When $K$ is sufficiently large, the distortion rate function is bounded
by \begin{equation}
\frac{t}{t+2}\left(c_{n,p,\beta}K\right)^{-\frac{2}{t}}\lesssim D^{*}\left(K\right)\lesssim\frac{2\Gamma\left(\frac{2}{t}\right)}{t}\left(c_{n,p,\beta}K\right)^{-\frac{2}{t}}.\label{eq:DRF_bounds}\end{equation}

\end{thm}

Due to the length limit, we only sketch the proof here. The lower
bound is proved by an optimization argument. The key is to construct
an ideal quantizer, which may not exist, to minimize the distortion.
Suppose that there exists $K$ metric balls of the same radius $\delta_{0}$
covering the whole $\mathcal{G}_{n,p}\left(\mathbb{L}\right)$ completely
without any overlap. Then the quantizer which maps each of those balls
into its center gives the minimum distortion among all quantizers.
Of course such a covering may not exist, provding a lower bound of
the distortion rate function.

The upper bound is derived by characterizing the average distortion
of the ensemble of random codes. Define a random code with size $K$
as $\mathcal{C}_{\mathrm{rand}}=\left\{ P_{1},P_{2},\cdots,P_{K}\right\} $
where $P_{i}$'s are independently drawn from the uniform distribution
on the $\mathcal{G}_{n,p}\left(\mathbb{L}\right)$. For any given
$Q\in\mathcal{G}_{n,p}\left(\mathbb{L}\right)$, define $X_{i}\triangleq d_{c}^{2}\left(P_{i},Q\right)$
and $W_{K}\triangleq\min\left(X_{1},\cdots,X_{K}\right)=\underset{P_{i}\in\mathcal{C}_{\mathrm{rand}}}{\min}\; d_{c}^{2}\left(P_{i},Q\right)$.
Since the codewords $P_{i}$'s $1\leq i\leq K$ are independently
drawn from the uniform distribution on the $\mathcal{G}_{n,p}\left(\mathbb{L}\right)$,
$X_{i}$'s $1\leq i\leq K$ are independent and identically distributed
(\emph{i.i.d.}) random variables with the cumulative distribution
function (CDF) given by Theorem \ref{thm:Volume_formula}. According
to $X_{i}$'s CDF, the CDF of $W_{K}$ can be calculated by extreme
order statistics. We prove that for any given $Q\in\mathcal{G}_{n,p}\left(\mathbb{L}\right)$,
$K^{\frac{t}{2}}\cdot\mathrm{E}_{W_{K}}\left[W_{K}\right]$ converges
to $\frac{2\Gamma\left(\frac{2}{t}\right)}{t}c_{n,p,\beta}^{-\frac{2}{t}}$
as $K$ approaches infinity. Thus, $K^{\frac{t}{2}}\cdot\mathrm{E}_{Q}\left[\mathrm{E}_{W_{K}}\left[W_{K}\right]\right]=K^{\frac{t}{2}}\cdot\mathrm{E}_{\mathcal{C}_{\mathrm{rand}}}\left[D\left(\mathcal{C}_{\mathrm{rand}}\right)\right]$
converges to the same constant, providing an upper bound of $D^{*}\left(K\right)$.

It is worthy to point out that since the upper bound is corresponding
to the ensemble of random codes, it is often used as an approximation
to the distortion rate function in practice.

Fig. \ref{cap:DRF_bounds} compares the simulated distortion rate
function with its lower bound and upper bound in (\ref{eq:DRF_bounds}).
To simulate the distortion rate function, we use the max-min criterion
\cite{Love_IT03_Grassman_Beamforming_MIMO} to design codes and use
the minimum distortion of the designed codes as the distortion rate
function. Simulations show that the bounds in (\ref{eq:DRF_bounds})
hold for large $K$. When $K$ is relatively small, the formula (\ref{eq:DRF_bounds})
can serve as good approximations to the distortion rate function as
well. In addition, we compare our bounds with the approximation (the
{}``x'' markers) derived in \cite{Heath_ICASSP05_Quantization_Grassmann_Manifold}.
While the approximation in \cite{Heath_ICASSP05_Quantization_Grassmann_Manifold}
works for the case that $n=10$ and $p=2$ but doesn't work when $n\leq8$
and $p=2$, the bounds in (\ref{eq:DRF_bounds}) hold for arbitrary
$n$ and $p$.

\begin{figure}
\includegraphics[clip,scale=0.5]{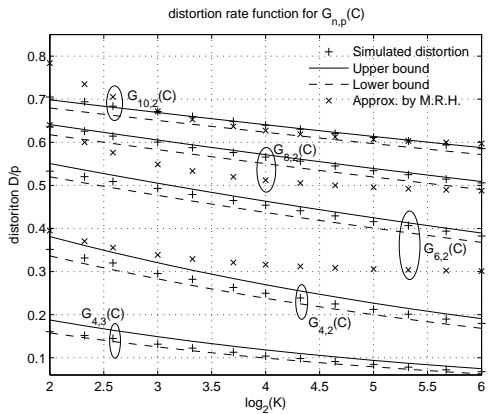}

\caption{\label{cap:DRF_bounds}Bounds on the distortion rate function}
\end{figure}

\section{Application to MIMO Systems with Finite Rate Channel State Feedback\label{sec:Application-to-MIMO}}





As an application of the derived quantization bounds on the Grassmann
manifold, this section discusses the information theoretical benefit
of finite rate channel state feedback for MIMO systems using power
on/off strategy. We will show that the benefit of the channel state
feedback can be accurately characterized by the distortion of a quantization
on the Grassmann manifold. 

The effect of finite rate feedback on MIMO systems using power on/off
strategy has been widely studied. MIMO systems with only one on-beam
are discussed in \cite{Sabharwal_IT03_Beamforming_MIMO,Love_IT03_Grassman_Beamforming_MIMO},
where the performance analysis is derived by geometric arguments in
the $\mathcal{G}_{n,1}\left(\mathbb{C}\right)$. For MIMO systems
with multiple on-beams, many works, e.g. \cite{Heath_ICASSP05_Quantization_Grassmann_Manifold,Love_SP05_Limited_feedback_unitary_precoding,Love_GlobeComm03_Limited_Feedback_Multiplexing},
employ Barg's formula (\ref{eq:Barg-formula}) for performance analysis,
which is only valid for MIMO systems with asymptotically large number
of antennas but fixed number of receive antennas. Valid for arbitrary
MIMO systems, the loss in information rate is quantified for high
SNR region in \cite{Rao_icc05_MIMO_spatial_multiplexing_limit_feedback},
which is hard to be generalized to other SNR regions. For all SNR
regimes, a formula to calculate the information rate is proposed in
\cite{Honig_MiliComm03_Asymptotic_MIMO_Limited_Feedback} by letting
the numbers of transmit and receive antennas and feedback rate approach
infinity simultaneously. But this formula overestimates the performance
in general. 

The system model of a wireless communication system with $L_{T}$
transmit antennas, $L_{R}$ receive antennas and finite rate channel
state feedback is given in Fig. \ref{cap:System-model}. The information
bit stream is encoded into the Gaussian signal vector $\mathbf{X}\in\mathbb{C}^{s\times1}$
and then multiplied by the beamforming matrix $\mathbf{P}\in\mathbb{C}^{L_{T}\times s}$
to generate the transmitted signal $\mathbf{T}=\mathbf{PX}$, where
$s$ is the dimension of the signal $\mathbf{X}$ satisfying $1\leq s\leq L_{T}$
and the beamforming matrix $\mathbf{P}$ satisfies $\mathbf{P}^{\dagger}\mathbf{P}=\mathbf{I}_{s}$.
In power on/off strategy, $\mathrm{E}\left[\mathbf{X}\mathbf{X}^{\dagger}\right]=P_{\mathrm{on}}\mathbf{I}_{s}$
where $P_{\mathrm{on}}$ is a positive constant to denote the on-power.
Assume that the channel $\mathbf{H}$ is Rayleigh flat fading, i.e.,
the entries of $\mathbf{H}$ are independent and identically distributed
(i.i.d.) circularly symmetric complex Gaussian variables with zero
mean and unit variance ($\mathcal{CN}\left(0,1\right)$) and $\mathbf{H}$
is i.i.d. for each channel use. Let $\mathbf{Y}\in\mathbb{C}^{L_{R}\times1}$
be the received signal and $\mathbf{W}\in\mathbb{C}^{L_{R}\times1}$
be the Gaussian noise, then\[
\mathbf{Y}=\mathbf{HPX}+\mathbf{W},\]
where $E\left[\mathbf{W}\mathbf{W}^{\dagger}\right]=\mathbf{I}_{L_{R}}$.
We also assume that there is a beamforming codebook $\mathcal{B}=\left\{ \mathbf{P}_{i}\in\mathbb{C}^{L_{T}\times s}:\right.$
$\left.\mathbf{P}_{i}^{\dagger}\mathbf{P}_{i}=\mathbf{I}_{s}\right\} $
declared to both the transmitter and the receiver before the transmission.
At the beginning of each channel use, the channel state $\mathbf{H}$
is perfectly estimated at the receiver. A message, which is a function
of the channel state, is sent back to the transmitter through a feedback
channel. The feedback is error-free and rate limited. According to
the channel state feedback, the transmitter chooses an appropriate
beamforming matrix $\mathbf{P}_{i}\in\mathcal{B}$. Let the feedback
rate be $R_{\mathrm{fb}}$bits/channel use. Then the size of the beamforming
codebook $\left|\mathcal{B}\right|\leq2^{R_{\mathrm{fb}}}$. The feedback
function is a mapping from the set of channel state into the beamforming
matrix index set, $\varphi:\;\left\{ \mathbf{H}\right\} \rightarrow\left\{ i:\;1\leq i\leq\left|\mathcal{B}\right|\right\} $.
This section will quantify the corresponding information rate \[
\mathcal{I}=\underset{\mathcal{B}:\left|\mathcal{B}\right|\leq2^{R_{\mathrm{fb}}}}{\max}\underset{\varphi}{\max}\;\mathrm{E}\left[\log\left|\mathbf{I}_{L_{R}}+P_{\mathrm{on}}\mathbf{H}\mathbf{P}_{\varphi\left(\mathbf{H}\right)}\mathbf{P}_{\varphi\left(\mathbf{H}\right)}^{\dagger}\mathbf{H}\right|\right],\]
where $P_{\mathrm{on}}=\rho/s$ and $\rho$ is the average received
SNR.

\begin{figure}
\includegraphics[clip,scale=0.55]{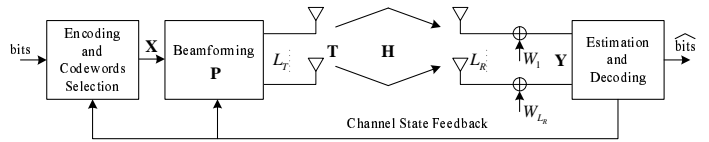}

\caption{\label{cap:System-model}System model}
\end{figure}

Before discussing the finite rate feedback case, we consider the case
that the transmitter has full knowledge of the channel state $\mathbf{H}$.
In this setting, the optimal beamforming matrix is given by $\mathbf{P}_{\mathrm{opt}}=\mathbf{V}_{s}$
where $\mathbf{V}_{s}\in\mathbb{C}^{L_{T}\times s}$ is the matrix
composed by the right singular vectors of $\mathbf{H}$ corresponding
to the largest $s$ singular values \cite{Dai_05_Power_onoff_strategy_design_finite_rate_feedback}.
The corresponding information rate is \begin{equation}
\mathcal{I}_{\mathrm{opt}}=\mathrm{E}_{\mathbf{H}}\left[\sum_{i=1}^{s}\mathrm{ln}\left(1+P_{\mathrm{on}}\lambda_{i}\right)\right],\label{eq:I_perfect_beamforming}\end{equation}
where $\lambda_{i}$ is the $i^{\mathrm{th}}$ largest eigenvalue
of $\mathbf{H}\mathbf{H}^{\dagger}$. In \cite{Dai_05_Power_onoff_strategy_design_finite_rate_feedback},
we derive an asymptotic formula to approximate a quantity of the form
$\mathrm{E}_{\mathbf{H}}\left[\sum_{i=1}^{s}\ln\left(1+c\lambda_{i}\right)\right]$
where $c>0$ is a constant. Apply the asymptotic formula in \cite{Dai_05_Power_onoff_strategy_design_finite_rate_feedback}.
$\mathcal{I}_{\mathrm{opt}}$ can be well approximated.

The effect of finite rate feedback can be characterized by the quantization
bounds in the Grassmann manifold. For finite rate feedback, we define
a suboptimal feedback function\begin{equation}
i=\varphi\left(\mathbf{H}\right)\triangleq\underset{1\leq i\leq\left|\mathcal{B}\right|}{\arg\ \min}\; d_{c}^{2}\left(\mathcal{P}\left(\mathbf{P}_{i}\right),\mathcal{P}\left(\mathbf{V}_{s}\right)\right),\label{eq:feedback-fn-suboptimal}\end{equation}
where $\mathcal{P}\left(\mathbf{P}_{i}\right)$ and $\mathcal{P}\left(\mathbf{V}_{s}\right)$
are the planes in the $\mathcal{G}_{L_{T},s}\left(\mathbb{C}\right)$
generated by $\mathbf{P}_{i}$ and $\mathbf{V}_{s}$ respectively.
In \cite{Dai_05_Power_onoff_strategy_design_finite_rate_feedback},
we show that this feedback function is asymptotically optimal as $R_{\mathrm{fb}}\rightarrow+\infty$
and near optimal when $R_{\mathrm{fb}}<+\infty$. With this feedback
function and assuming that the feedback rate $R_{\mathrm{fb}}$ is
large, it has been shown in \cite{Dai_05_Power_onoff_strategy_design_finite_rate_feedback}
that \begin{eqnarray}
\mathcal{I} & \approx & \mathrm{E}_{\mathbf{H}}\left[\sum_{i=1}^{s}\ln\left(1+\eta_{\sup}P_{\mathrm{on}}\lambda_{i}\right)\right],\label{eq:I_finite_feedback}\end{eqnarray}
where \begin{eqnarray}
\eta_{\sup} & \triangleq & 1-\frac{1}{s}\underset{\mathcal{B}:\left|\mathcal{B}\right|\leq2^{R_{\mathrm{fb}}}}{\inf}\;\mathrm{E}_{\mathbf{V}_{s}}\left[\underset{1\leq i\leq\left|\mathcal{B}\right|}{\ \min}\; d_{c}^{2}\left(\mathcal{P}\left(\mathbf{P}_{i}\right),\mathcal{P}\left(\mathbf{V}_{s}\right)\right)\right]\nonumber \\
 & = & 1-\frac{1}{s}D^{*}\left(2^{R_{\mathrm{fb}}}\right).\label{eq:PEF-DRF}\end{eqnarray}
Thus, the difference between perfect beamforming case (\ref{eq:I_perfect_beamforming})
and finite rate feedback case (\ref{eq:I_finite_feedback}) is quantified
by $\eta_{\sup}$, which depends on the distortion rate function on
the $\mathcal{G}_{L_{T},s}\left(\mathbb{C}\right)$. Substitute quantization
bounds (\ref{eq:DRF_bounds}) into (\ref{eq:PEF-DRF}) and apply the
asymptotic formula in \cite{Dai_05_Power_onoff_strategy_design_finite_rate_feedback}
for $\mathrm{E}_{\mathbf{H}}\left[\sum_{i=1}^{s}\ln\left(1+c\lambda_{i}\right)\right]$.
Approximations to the information rate $\mathcal{I}$ are derived
as functions of the feedback rate $R_{\mathrm{fb}}$. 

\begin{figure}
\includegraphics[clip,scale=0.5]{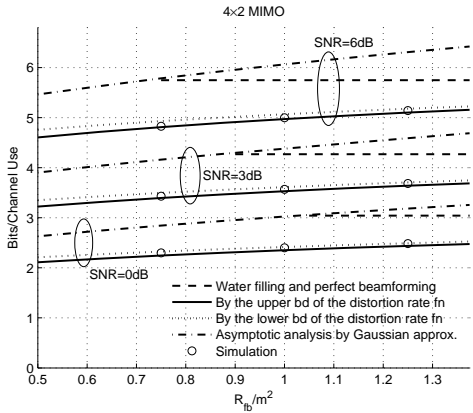}

\caption{\label{cap:Fig-Honig}Performance of finite size beamforming codebook.}
\end{figure}

Simulations verify the above approximations. Let $m=\min\left(L_{T},L_{R}\right)$.
Fig. \ref{cap:Fig-Honig} compares the simulated information rate
(circles) and approximations as functions of $R_{\mathrm{fb}}/m^{2}$.
The information rate approximated by the lower bound (solid lines)
and the upper bound (dotted lines) in (\ref{eq:DRF_bounds}) are presented.
As a comparison, we also include another performance approximation
(dash-dot lines) proposed in \cite{Honig_MiliComm03_Asymptotic_MIMO_Limited_Feedback},
which is based on asymptotic analysis and Gaussian approximation.
The simulation results show that the performances approximated by
the bounds (\ref{eq:DRF_bounds}) match the actual performance almost
perfectly and are much more accurate than the one in \cite{Honig_MiliComm03_Asymptotic_MIMO_Limited_Feedback}.

\section{Conclusion\label{sec:Conclusion}}

This paper considers the quantization problem on the Grassmann manifold.
Based on the explicit volume formula for a metric ball in the $\mathcal{G}_{n,p}\left(\mathbb{L}\right)$,
the corresponding Gilbert-Varshamov and Hamming bounds are obtained.
Assuming the uniform source distribution and the distortion defined
by the squared chordal distance, the distortion rate function is characterized
by establishing tight lower and upper bounds. As an application of
these results, the information rate of a MIMO system with finite rate
channel state feedback is accurately quantified for abitrary finite
number of antennas for the first time.

\bibliographystyle{IEEEtran}
\bibliography{bib/_Blum,bib/_Heath,bib/_Liu_Dai,bib/_love,bib/_Rao,bib/_Tse,bib/FeedbackMIMO_append,bib/MIMO_basic,bib/RandomMatrix}

\end{document}